\documentclass[12pt]{article}
\usepackage{sectsty}
\usepackage[T1]{fontenc}
\usepackage{kpfonts}
\usepackage[dvipsnames]{xcolor}
\usepackage{pdfpages}
\usepackage{accents}
\usepackage{tikz-cd}
\usepackage{float}
\usepackage[round]{natbib}
\usepackage{multirow}
\usepackage{dutchcal}
\usepackage{pdfpages}
\usepackage{bbm}
\usepackage{sgame}
\usepackage{mathtools}
\usepackage{amsthm}
\usepackage{enumitem}
\usepackage[margin=1in]{geometry}
\usepackage{caption}
\usepackage{subcaption}
\usepackage{epigraph}
\usepackage{listings}
\usetikzlibrary{calc}
\usetikzlibrary{shapes,arrows}
\usepackage{tcolorbox}
\usepackage{marvosym}

\newcommand{\ubar}[1]{\underaccent{\bar}{#1}}

\newtheorem{theorem}{Theorem}[section]
\newtheorem{proposition}[theorem]{Proposition}
\newtheorem{lemma}[theorem]{Lemma}
\newtheorem{corollary}[theorem]{Corollary}
\newtheorem{claim}[theorem]{Claim}

\theoremstyle{definition}

\sectionfont{\color{MidnightBlue}}
\subsectionfont{\color{MidnightBlue}}

\usepackage{hyperref}
\hypersetup{
    colorlinks=true,
    linkcolor=OrangeRed,
    filecolor=Periwinkle,      
    urlcolor=Periwinkle,
    citecolor=Periwinkle,
}

\definecolor{backcolour}{rgb}{0.63, 0.79, 0.95}

\lstdefinestyle{mystyle}{
  backgroundcolor=\color{backcolour},
  basicstyle=\ttfamily\footnotesize,
  breakatwhitespace=false,         
  breaklines=true,                 
  captionpos=b,                    
  keepspaces=true,                 
  numbers=left,                    
  numbersep=5pt,                  
  showspaces=false,                
  showstringspaces=false,
  showtabs=false,                  
  tabsize=2
}


\usepackage{thmtools,thm-restate}

\begin{document}
\author{Brian C. Albrecht\thanks{International Center for Law and Economics. Email: \href{mailto:mail@briancalbrecht.com}{mail@briancalbrecht.com}} \and Mark Whitmeyer\thanks{Arizona State University. Email: \href{mailto:mark.whitmeyer@gmail.com}{mark.whitmeyer@gmail.com}\newline We thank Andy Kleiner, Jiangtao Li, David Rahman, Doron Ravid, Jingyi Xue, Kun Zhang, Jidong Zhou, and seminar and conference audiences at the 2021 SEA Conference and Singapore Mangement University for their comments.} }
\title{Comparison Shopping: Learning Before Buying From Duopolists}
\maketitle

\begin{abstract}
We explore a model of duopolistic competition in which consumers learn about the fit of each competitor's product. In equilibrium, consumers comparison shop: they learn only about the \textit{relative} values of the products. When information is cheap, increasing the cost of information decreases consumer welfare; but when information is expensive, this relationship flips. As information frictions vanish, there is a limiting equilibrium that is \textit{ex post} efficient.
\end{abstract}

\newpage

\section{Introduction}


Many markets for consumer goods and services share the following features. First, there are multiple competitors in the market. Second, the firms' products are horizontally differentiated, but how well they suit any given consumer's taste is not clear \textit{ex ante}. Third, there is a lot of information available to consumers about the products--via product specifications, reviews, or expert evaluations--but accessing or parsing it is costly to consumers. Fourth, the sellers are not ignorant of these details; they set prices while taking into explicit account that they compete with other sellers and that consumers are learning about the products. The market for new cars fits this description well; as do the markets for appliances, electronics, and over-the-counter (OTC) assets. Likewise, contractors, pet-sitters, and mechanics all participate in markets with these characteristics.


We study a stylized model that aims to capture these essential features. Our model incorporates flexible, costly learning by a buyer and competition between two sellers with horizontally differentiated products that take binary values. This buyer's ability to learn is extremely flexible: she can acquire \textit{any} signal about her valuations. Signals are costly, but we assume this cost is smooth and strictly increasing in how informative the signal is, and the sellers set prices without learning the buyer's strategy or signal realization. This last feature is especially realistic in markets for services, where the idiosyncrasy of a buyer's need means that the buyer is unaware of what a contractor's (or mechanic's) price quote will be when she is researching them. 

\cite*{ravid2020learning} \hypertarget{RRS}{(RRS)} study a monopoly version of this problem and therein expose the troubling result that with ``learning before trading,'' as information costs go to zero, the limit equilibrium is \textit{ex post} inefficient.\footnote{\hyperlink{RRS}{RRS} assume that the consumer's valuation is distributed according to some prior with a density. As we argue in \(\S\)\ref{sec:rrsbinary}, their inefficiency result persists if the consumer's valuation is binary, as in our paper.} To be more precise, they show that as information costs vanish, the equilibrium converges to one in which the consumer does not purchase from the seller with positive probability, despite the fact that trade is always efficient. Moreover, with positive probability, the consumer does not purchase even though her match value (\textit{ex post}) is higher than the monopolist's price. 

\medskip
In this paper, we prove three main results:

\medskip

\noindent \textbf{1. Comparison shopping:} We prove that in equilibrium, the buyer chooses learning strategies that reveal only the relative value of the firms' products, which we call "comparison shopping." More specifically, if \((x,y)\) is the consumer's valuation for goods from seller one and seller two, she learns only along the line \(y= 1 - x\). As information becomes cheaper, the consumer acquires posteriors further and further from the prior. Eventually, the prior constrains this, which causes the consumer to ``leave'' beliefs at the prior: with strictly positive probability, the consumer is indifferent between the two firms.

\medskip

\noindent \textbf{2. Novel benefit to competition:} We prove that competition partially overturns the \textit{ex post} inefficiency found in the monopoly model of \hyperlink{RRS}{RRS}. We find that as the information costs vanish, there is a limiting equilibrium that is \textit{ex post} efficient.



\medskip

\noindent \textbf{3. Ambiguous effect of frictions on consumer welfare:} We also ask when cheaper information increases consumer welfare. Even without the cost of acquiring information, there is a trade-off for the consumer.\footnote{Many papers explore this trade-off, e.g.,  \cite*{moscarini2001} and \cite*{al}.} If she acquires more information, she has a better match quality with the seller, but that softens competition between the sellers. If she remains ignorant, she can induce stronger price competition, but her lack of information hurts her purchase decision.  When the cost of information is high, consumer welfare is decreasing in the cost of information, but
when information is cheap, this relationship flips. 

The rest of the paper is as follows: \(\S\)\ref{sec:related_work} covers related work; \(\S\)\ref{sec:model} sets up the model; \(\S\)\ref{sec:rrsbinary} presents the binary-state version of \hyperlink{RRS}{RRS} as another benchmark; \(\S\)\ref{sec:pricing} solves the optimal pricing game, conjecturing the optimal learning; \(\S\)\ref{sec:learning} solves for optimal learning, given the equilibrium pricing, and proves the main results; and \(\S\)\ref{extensions} discusses our assumptions, illustrates how our results extend when the consumer's prior has a density (\(\S\)\ref{sec:density}) and sets up the benchmark where sellers observe what the buyers learn (\(\S\)\ref{sec:observable}).

\subsection{Related Work}
\label{sec:related_work}

Beyond the closest paper to ours, \hyperlink{RRS}{RRS}, there is now a sizeable collection of papers that study the important question, \textit{``what can happen in markets under different information structures?''} \cite*{searchinfoprices} asks this in the context of a search market, and \cite*{limitsofprice} characterizes possible market outcomes when there is a single monopolistic seller. \cite*{cournot} look at Cournot competition through this lens, while \cite*{Roes} and \cite*{Condor} study consumer-optimal information and distributions over valuations, respectively. \cite*{za}, \cite*{al}, \cite*{shizhang}, \cite*{elliot}, \cite*{vickzhou}, \cite*{armvick}, and \cite*{hu} all study variations on this theme in markets with imperfect competition. Of special note is \cite*{moscarini2001}, who study price competition by duopolists facing a privately informed buyer. Their consumer's information is exogenous, and they study the pricing-only game between the firms. 

There are a number of papers that explore information acquisition in markets by consumers. Most limit their analysis to the monopolist scenario. \cite*{branco12}, 
 \cite*{branco2016too}, \cite*{martin2017} (in which, notably, the seller is informed), \cite*{pease2018shopping}, and \cite*{lang2019try} all look at how a monopolist sells to consumers who may subsequently acquire information about their valuation for the product. Importantly, in these papers, the consumer observes the firm's price before deciding how and what to learn. This timing is also assumed in the oligopolistic setting of \cite*{matvejka2012simple}.\footnote{The literature on information acquisition in auctions/mechanism design (e.g., \cite*{persico}, \cite*{shi12}, \cite*{kimkoh}, \cite*{mensch}, and \cite*{thereze}) also imposes that the mechanism is publicized--and committed to by the designer--before consumers acquire information.} 
 
 In a contemporaneous paper, \cite*{biglaiser2023price} study a similar model in which duopolistic firms set prices and consumers learn simultaneously. The precise details of their model and the focus of the paper differ from ours. For example, their setting is single-dimensional by construction: a consumer has value \(1\) for one of the firms and \(0\) for the other, whereas we place our analysis in the unit square. This multidimensionality is key to our efficiency result: we would not obtain \textit{ex post} efficiency in the free-information limit if the state space were single-dimensional. They do not study this efficiency question but instead center their analysis around product differentiation and present a nice application to platform design.

 \cite*{jainwhitflex} explore flexible information acquisition by consumers in a large oligopolistic market with search frictions \'{a} la \cite*{asher}. Like this paper, the primary focus there is the case in which the consumer acquires information before observing a firm's price offer. \cite*{ravid20} studies a bilateral bargaining scenario in which the seller's price offer is made prior to the consumer's learning but the consumer must incur a cost to learn about the seller's offer as well. \cite*{pieroth22} extend this model (learning about the qualities and price offers) to a duopoly setting.

Finally, our equilibrium construction requires solving a multidimensional information design problem. To accomplish this, we use recent technical results and insights from \cite*{persuasionduality}, \cite*{yoder2021}, and \cite*{kleiner2023extreme}.

\section{Model}
\label{sec:model}

There are \(2\) \textit{ex ante} identical horizontally differentiated firms, indexed by \(i\). Each of them is selling a product whose value to a representative consumer is an identically distributed binary random variable \(Z_i\) with support on \(\left\{0,1\right\}\) and \(\mathbb{P}\left(Z_i = 1\right) = \frac{1}{2}\) for \(i = 1, 2\). We also assume that the joint distribution of the random variables is symmetric: \(\mathbb{P}\left(Z_1 = \left.1\right| Z_2 = v\right) = \mathbb{P}\left(Z_2 = \left.1\right| Z_1 = v\right)\) for all \(v \in \left\{0,1\right\}\). We denote \(\omega \coloneqq \mathbb{P}\left(Z_2 = 1, Z_1 = 0\right)\) and assume a \textbf{Full support prior} with \(0 < \omega \leq \frac{2}{5}\). Note that this condition is satisfied by the i.i.d. case. In \(\S\) \ref{sec:density}, we allow for a continuous (mean \(\frac{1}{2}\)) prior.

The consumer privately learns about the state of the world at a cost, formalized as follows. Let \(\mathcal{F}\) denote the set of all distributions supported on \(\left[0,1\right]^{2}\) that are fusions of the prior, i.e., that can be obtained by observing some signal. The consumer may acquire any fusion \(F \in \mathcal{F}\) at cost \(C\colon \mathcal{F} \to \mathbb{R}\), where \(C\) satisfies the following assumptions:
\begin{tcolorbox}[colframe=Thistle,colback=white]\textbf{Assumptions on the Cost Functional:} We assume for any \(F \in \mathcal{F}\),
\[C\left(F\right) = \kappa \int c dF \text{ ,}\]
for some strictly convex, thrice differentiable, function \(c\) of the form
\[c\left(x,y\right) = \varphi\left(x\right) + \varphi\left(1-x\right) + \varphi\left(y\right) + \varphi\left(1-y\right)\text{,}\]
and scalar \(\kappa > 0\); with \(\varphi\left(
\frac{1}{2}\right) = 0\), \(\varphi\left(z\right) < \infty\) for all \(z \in \left(0,1\right)\), and \(\lim_{z' \uparrow 1}\left|\varphi'\left(z'\right)\right| = \lim_{z' \downarrow 0}\left|\varphi'\left(z'\right)\right| = \infty\). We also make the technical assumption that \(c_{xxx} \leq 0\).
\end{tcolorbox}
The following function (which we use for any figures) fits our specification:
\[\label{merc}\tag{\Mercury}
    c\left(x,y\right) = x \log x + \left(1-x\right)\log\left(1-x\right) + y\log y + \left(1-y\right)\log\left(1-y\right) - 2 \log \frac{1}{2}\text{.}\]
We also specify that the consumer's utility is additively separable in her value for the good she purchases, its price, and her cost of acquiring information: if she purchases a product with expected value \(x\) at price \(p\) and at posterior \(\left(x,y\right)\), her utility is \(x - p-\kappa c\left(x,y\right)\). We specify \textbf{Full Market Converage}: the consumer has a negligible (or nonexistent) outside option and so will always purchase from one of the firms. For simplicity, we set the marginal costs of production for the two firms to \(0\).

The timing of the game is straightforward: 
\begin{enumerate}[label={(\roman*)},noitemsep,topsep=0pt]
    \item \textbf{Private Learning:} The consumer acquires information about the two products. Neither firm observes this learning.
    \item \textbf{Simultaneous Price Setting:} The firms simultaneously post prices.
    \item \textbf{Purchase Decision} Given posterior value \(\left(x,y\right)\) and prices \(\left(p_1,p_2\right)\), the consumer purchases from Firm \(1\) (\(2\)) if \(x-p_1 > \ (<) \ y-p_2\).
\end{enumerate}
In the first (information acquisition) stage, the consumer solves
\[\max_{F \in \mathcal{F}} \int \left(u-\kappa c\right)dF \text{,}\]
where \(u\colon \left[0,1\right]^{2} \to \mathbb{R}\) is the consumer's reduced form utility from posterior \(\left(x,y\right)\).

\section{The Monopolist Scenario with Binary States}\label{sec:rrsbinary}

Before moving to duopoly sellers, let us go through the binary-state version of the single-seller setting of \hyperlink{RRS}{RRS} to allow for an explicit comparison. The buyer has value \(1\) for the seller's good with probability \(\mu \in \left(0,1\right)\); otherwise, her value is \(0\). We translate our assumptions on the cost functional to the single-firm environment in the natural manner; in particular, now \(\lim_{x \downarrow 0} \left|c\left(x\right)\right|= \lim_{x \uparrow 1} \left|c\left(x\right)\right| = \infty\). Following \hyperlink{RRS}{RRS}, we assume that the consumer has an outside option of \(0\).

To characterize the equilibrium, we define the following distribution of posteriors
\[F\left(x\right) = 1 - \frac{\ubar{x}}{x}, \ \text{on} \ \left[\ubar{x},\bar{x}\right]\text{,}\]
and distribution of prices
\[G\left(p\right) = \kappa \left(c'\left(p\right) - c'\left(\ubar{x}\right)\right), \ \text{on} \ \left[\ubar{x},\bar{x}\right]\text{,}\]
where \(\ubar{x}\) and \(\bar{x}\) are pinned down by two equations. As \(G\) is a CDF, we have
\[\label{eq1}\tag{\(M_1\)}1 = \kappa \left(c'\left(\bar{x}\right) - c'\left(\ubar{x}\right)\right)\text{.}\]
We also have a Bayes-plausibility condition
\[\left(1-F\left(\bar{x}\right)\right)\bar{x} + \int_{\ubar{x}}^{\bar{x}}xdF\left(x\right) = \mu\text{.}\]
After substituting in the functional form of \(F\) and simplifying, we have the the second equation we need to pin down  \(\ubar{x}\) and \(\bar{x}\):
\[\label{eq2}\tag{\(M_2\)}1 + \log{\frac{\bar{x}}{\ubar{x}}} - \frac{\mu}{\ubar{x}} = 0\text{.}\]
We, therefore, have the following proposition, which specifies the equilibrium distributions of posteriors and prices: 
\begin{proposition}\label{monoppropproof}
    If \(\kappa > 0\), the unique equilibrium is for the consumer to randomize over posteriors according to distribution \(F\) given in \ref{eq1} and for the firm to randomize over prices according to distribution \(G\) given in \ref{eq2}.
\end{proposition}
\noindent  The proof is in \(\S\)\ref{proof:monopoly_eq}. We are interested in the equilibrium as \(\kappa\) goes to zero: 
\begin{proposition}\label{monopcsproof}
    As \(\kappa \downarrow 0\), the firm's pricing strategy converges to the degenerate distribution on price \(1\) and the consumer's distribution over valuations converges to the truncated Pareto distribution \(F\) specified above with \(\ubar{x} = a\) and \(\bar{x} = 1\), where \(a\) is the unique solution to \(1 - \log(a) - \frac{\mu}{a} = 0\) in \(\left(0,\mu\right)\).
\end{proposition}
The proof is in \(\S\)\ref{proof:monopoly_cor}. Crucially, as in the main setting of \hyperlink{RRS}{RRS}, even if the prior is binary, the limiting equilibrium outcome as frictions vanish is not \textit{ex post} efficient. With strictly positive probability, the consumer does not purchase when her true valuation is \(1\), as with strictly positive probability her posterior is some value in \(\left[a,1\right)\), which is below the price of \(1\).

\section{Two Pricing Games Between Firms}
\label{sec:pricing}

 Eventually, we will characterize the symmetric equilibria of the game with flexible learning by the consumer and price-setting by the firms. To do that, we solve the game by backward induction, and so we start by examining two games of pure price setting. In this ``second stage,'' we are holding fixed the consumer's learning; \textit{viz.}, fixing some conjectured distribution over consumer valuations for the two firms' products. With this conjecture fixed, we solve for the equilibrium pricing strategy. In \(\S\)\ref{sec:learning}, we then show these two potential distributions of beliefs are correct equilibrium distributions of valuations, given the consumer's correct conjectures of the firms' pricing strategies. 

\subsection{Value Distribution With Symmetric Two-Point Support}

Suppose the consumer has a symmetric distribution over valuations for the firms' products with support on \(\left(0,\lambda\right)\) and \(\left(\lambda, 0\right)\), each with probability \(\frac{1}{2}\), and where \(\lambda > 0\). Getting slightly ahead of ourselves, these beliefs arise from comparison shopping where the consumer learns only that she likes one of the goods more than the other by $\lambda$. This game is a modified, full-market-coverage, version of \cite*{moscarini2001}. By standard, undercutting arguments with finite valuation types, firms must randomize in any equilibrium.
\begin{lemma}\label{eq_rand_2types}
There exist no symmetric equilibria in which firms do not randomize over prices. 
\end{lemma}

\noindent The proof is in \(\S\)\ref{proof:eq_rand_2types}. As a result, we search for an equilibrium in which firms randomize over prices. We define the following piecewise distribution of prices, which is the equilibrium distribution:
\[\label{2firmeqpricing}\tag{\Earth}\Gamma\left(p\right) = \begin{cases}
\Gamma_L \coloneqq \frac{p - \sqrt{2}\lambda}{\lambda + p}, \quad &\sqrt{2}\lambda \leq p \leq \left(1+\sqrt{2}\right)\lambda\\
\Gamma_H \coloneqq \frac{\left(3+\sqrt{2}\right)\lambda - 2 p}{\lambda - p}, \quad &\left(1+\sqrt{2}\right)\lambda \leq p \leq \left(2+\sqrt{2}\right)\lambda\\
\end{cases} \text{ .}\]

\begin{proposition}\label{eqdist}
In the price-setting game of this subsection, the unique symmetric equilibrium is for each firm to choose the distribution over prices \(\Gamma\) specified in Expression \ref{2firmeqpricing}.
\end{proposition}

The proof is in \(\S\)\ref{proof:eqdist}. This equilibrium distribution arises as the distribution that generates unit-elastic demand for the other firm (given the conjectured uncertainty), rendering it willing to randomize.

\subsection{Value Distribution With Symmetric Three-Point Support}

Now suppose the consumer has a symmetric distribution over valuations for the firms' products with support on \(3\) points as follows. For \(2\) points, her valuation for one of the firms is \(\lambda > 0\) greater than that of the other firm, just as in the previous subsection. At the \(3\)rd point, the consumer is indifferent between each of the firms.

After normalization, we specify that with probability \(q \leq \frac{2}{5}\) the consumer's vector of valuations for the \(2\) firms is \(\left(0, \lambda\right)\) and with probability \(1-2 q\) the consumer's vector of valuations is \((0,0)\). As in the case with two-point support, any equilibrium must involve randomization over prices.

\begin{lemma}\label{eq_rand_3types}
There exist no symmetric equilibria in which firms do not randomize over prices. Moreover, firms' distributions over prices cannot have atoms.
\end{lemma}

The proof is in \(\S\)\ref{proof:eq_rand_3types}. Again, we search for an equilibrium in which firms randomize over prices. Define the distribution \(\Phi\left(p\right)\)
\[\tag{\Mars}\label{priceset2}\Phi\left(p\right) \coloneqq \frac{\left(1-q\right)\left(p\left(1-2q\right)-\lambda q\right)}{p\left(1-2q\right)^{2}} \quad \text{on} \quad \left[\frac{q}{1-2q}\lambda, \frac{q}{1-2q}\lambda + \lambda\right]  \text{.}\]
\begin{proposition}\label{eq_dist_3types}
    In the price-setting game of this subsection, it is an equilibrium for each firm to choose the distribution over prices \(\Phi\) specified in Expression \ref{priceset2}.
\end{proposition}
\noindent The proof is in \(\S\)\ref{proof:eq_dist_3types}.

\section{Consumer Learning}
\label{sec:learning}

\begin{figure}
    \centering
    \includegraphics[scale=1]{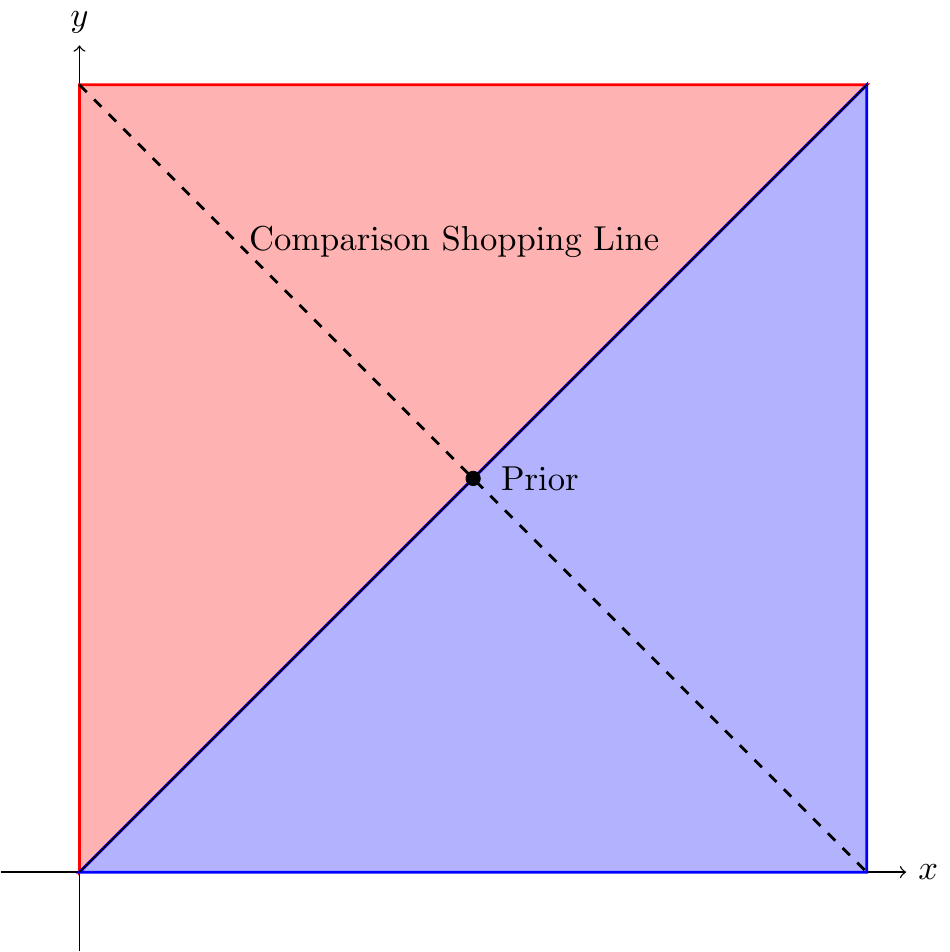}
    \caption{The Space of Valuations}
    \label{fig0}
\end{figure}

In this section, we take as given the price-setting by the firms that we characterized in \(\S\)\ref{sec:pricing} and find the consumer's optimal learning. Combined, the pricing and the learning will give us the equilibria in the grand game. 

First, note that if learning were free (\(\kappa = 0\)), the consumer would perfectly learn her valuation. However, the following theorem proves that as long as there exist frictions, no matter how small (\(\kappa > 0\)), the consumer only learns along the \textbf{Comparison Shopping} line \(y = 1 - x\). That is, the consumer's learning exclusively focuses on the \textit{relative} merits of each firm's product. Defining the set \(\ell^*\) as 
\[\ell^* \coloneqq \left\{\left(x,y\right) \in \left[0,1\right]^2 \ \colon \ y = 1 - x\right\}\text{,}\] and saying that the consumer \textbf{Comparison Shops} if her acquired distribution over posteriors is supported on a subset of \(\ell^*\). Figure \ref{fig0} illustrates the space of valuations and the comparison shopping line. Posteriors in the red 
(blue) region are those at which the consumer's valuation for firm \(2\)'s (\(1\)'s) good is highest. The dotted diagonal line is the comparison shopping line when the prior is the dot, which since the prior is symmetric is on the boundary of the red and blue regions. In equilibrium, the consumer only learns along the dotted, comparison shopping line. Our formal result is:
\begin{theorem}\label{thmcomparisonshop}
    If firms choose symmetric, atomless, distributions that admit densities with support on some closed interval \(\left[\ubar{p}, \bar{p}\right]\), the consumer comparison shops.
\end{theorem}
The proof is in \(\S\)\ref{proof:thmcomparisonshop}. The crucial observation behind this theorem is that the consumer's payoff as a function of her posterior is strictly concave along the vector \(\left(1, 1\right)\) and maximized along this vector by points on the comparison-shopping line. Bayes' plausibility then pins down the comparison shopping line \(1 - x\). 

\subsection{Solving the Information Acquisition Problem}

In solving the consumer's problem, we conjecture its solution and use the corresponding strategies by the firms in the pricing-only game to generate the consumer's value function. Then, we verify that the consumer's optimal learning is precisely the two- or three-point support that we conjectured in \(\S\)\ref{sec:pricing}. 

The value function for the consumer from acquiring belief \(\left(x,y\right)\) is 
\[\begin{split}
    V\left(x,y\right) \coloneqq &\mathbb{P}\left(x-p_1 \geq y-p_2\right) \mathbb{E}\left(x-\left.p_1\right|x-p_1 \geq y-p_2\right)\\
    &+ \mathbb{P}\left(y-p_2 \geq x-p_1\right) \mathbb{E}\left(y-\left.p_2\right|y-p_2 \geq x-p_1\right) - \kappa c\left(x,y\right) \text{,}
\end{split}\]
which is continuously differentiable except on \(\partial \left[0,1\right]^2\) (because firms randomize continuously over prices) and is bounded above on the entire square. Accordingly, by Theorem \(1\) of \cite*{persuasionduality}, we have weak duality and the price function solution lies weakly above the consumer's value in her information acquisition problem. From there, it is easy to solve the dual problem and verify that this corresponds to a solution to the primal problem. 

Intuitively, we can make use of the symmetry of the information acquisition problem and ``split'' the prior probabilities of \(\left(1,1\right)\) and \(\left(0,0\right)\) equally between the triangles
\[\label{delta1}\tag{\Jupiter}\Delta^1 \coloneqq \left\{\left(x,y\right) \in \left[0,1\right]^2 \colon 0 \leq x \leq 1 \ \text{\&} \ 0 \leq y \leq x\right\}\text{,}\]
and
\[\label{delta2}\tag{\Saturn}\Delta^2 \coloneqq \left\{\left(x,y\right) \in \left[0,1\right]^2 \colon 0 \leq x \leq 1 \ \text{\&} \ 1 \geq y \geq x\right\}\text{;}\]
before solving two standard \(3\)-state persuasion problems (as any simplex is homeomorphic to the standard simplex), on each of the two triangles. These problems satisfy the assumptions of \cite*{yoder2021}, whose Proposition 2 proves that the concavification approach is valid. It remains to verify that the maximum of the two concavifying planes is convex and lie every above the value function, and that the two planes either are the same or have \(y = x\) as their intersection.

The solution is then as follows. If information is expensive, the price function is just a single plane; i.e., the two concavifying planes are the same plane. If information is moderately expensive, the price function is the maximum of two planes that intersect at \(y = x\) and lie weakly above the value function on that line. Finally, if information is cheap, the price function is as in the moderate cost case, with the additional specification that it is equal to the value function at its minimum, along the line \(y = x\).

\subsection{Expensive Information}

\begin{figure}
\centering
\begin{subfigure}{.5\textwidth}
  \centering
  \includegraphics[scale=.18]{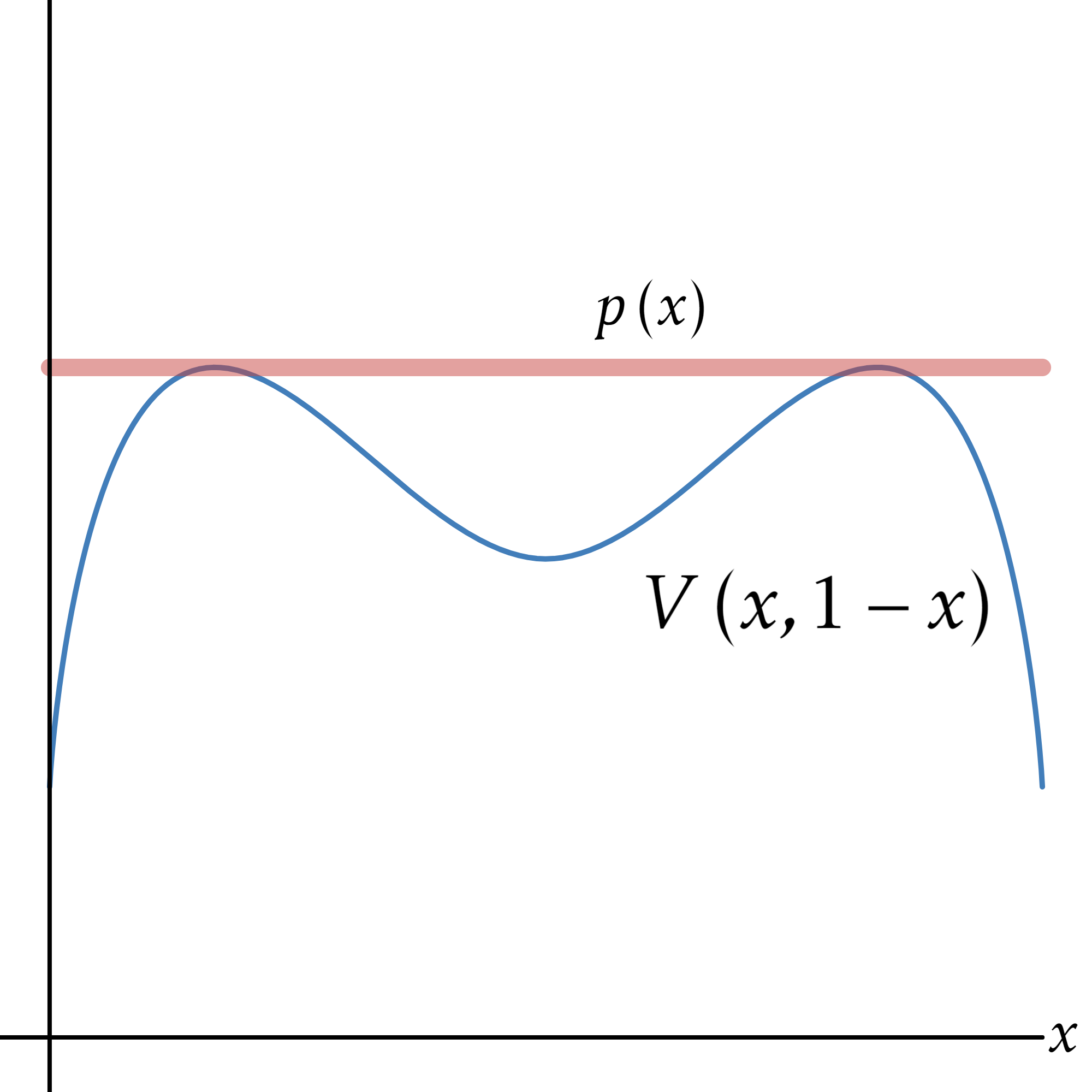}
  \caption{Expensive Information: \(\kappa \geq \bar{\kappa}\)}
  \label{figsub12}
\end{subfigure}%
\begin{subfigure}{.5\textwidth}
  \centering
  \includegraphics[scale=.18]{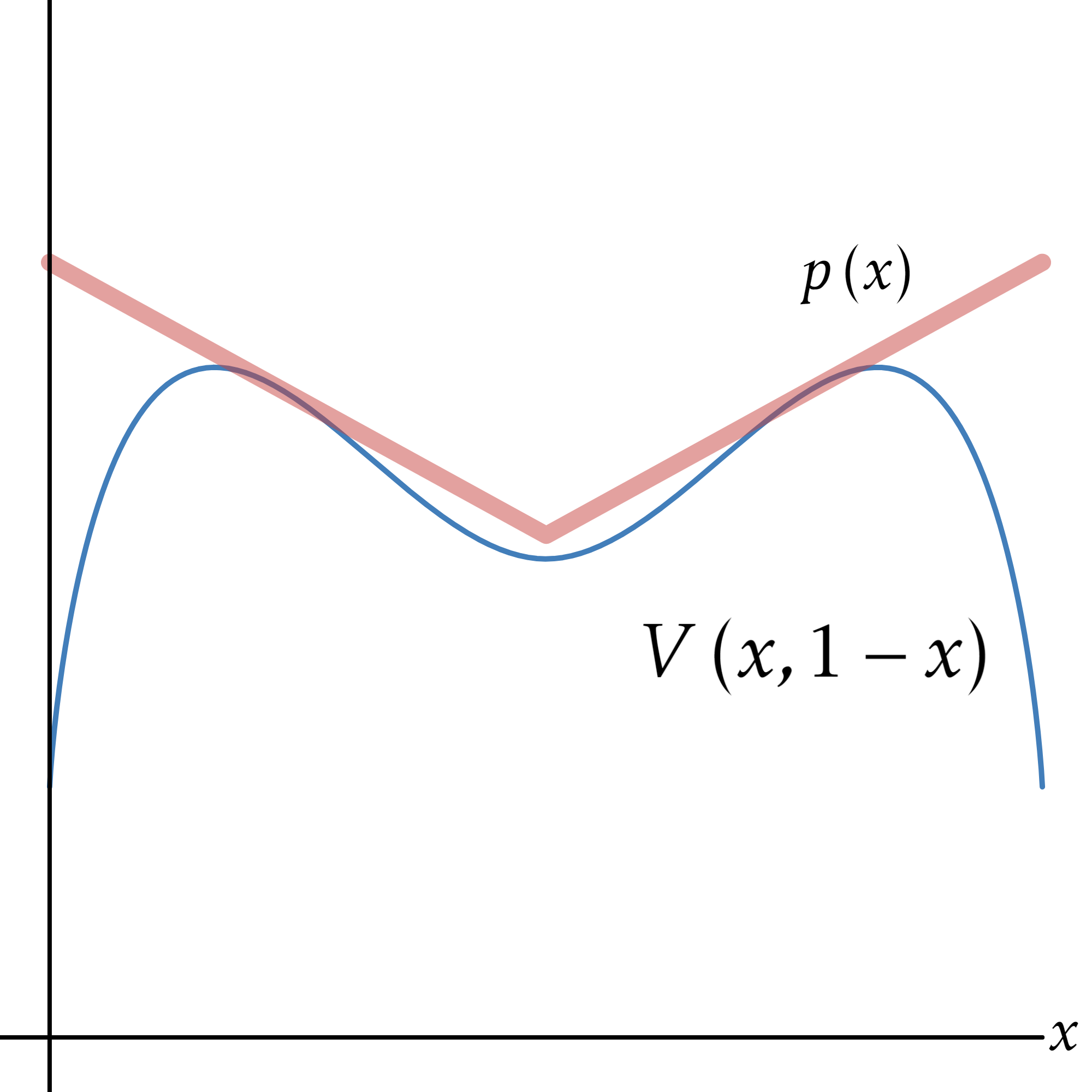}
  \caption{Moderately Costly Information: \(\kappa \in \left[\ubar{\kappa},\bar{\kappa}\right]\)}
  \label{figsub22}
\end{subfigure}
\caption{Two possible price functions when information is not cheap.}
\label{fig:notcheapinfo}
\end{figure}

The easiest scenario to analyze is that in which either information is expensive (\(\kappa\) is large). Our main result of this section is that if frictions are sufficiently large, then regardless of the prior, there is an equilibrium in which the consumer acquires a binary distribution over posteriors.

We say that a consumer \textbf{Comparison Shops With Uniform Two-point Support} if the consumer's acquired distribution over valuations is supported on \[\left\{\left(\frac{1-\lambda}{2}, \frac{1+\lambda}{2}\right), \left(\frac{1+\lambda}{2}, \frac{1-\lambda}{2}\right)\right\}\text{,}\] each with probability \(\frac{1}{2}\).

\begin{theorem}\label{firstmaintheorem}
    If \(\kappa\) is sufficiently high, there is an equilibrium in which the consumer comparison shops with uniform two-point support and firms randomize over prices according to Expression \ref{2firmeqpricing}.
\end{theorem}

The proof is in \(\S\)\ref{proof:firstmaintheorem}. Notably, when \(\kappa\) is sufficiently large, \(\lambda\) is strictly increasing in \(\kappa\): as frictions shrink, the consumer learns more and more in a mean-preserving spread sense. For all such \(\kappa\), the price function is a single plane with zero slope. Eventually (as \(\kappa\) continues to shrink), unless at most one product has high value, \(\kappa\) hits a threshold \(\bar{\kappa}\). Then, for all \(\kappa\) within some interval \(\left[\ubar{\kappa},\bar{\kappa}\right]\) the consumer's learning is the same--as are the pricing strategies by the firms. Here, the price function is the maximum of two planes whose intersection is the line \(y=x\). Both cases are depicted in Figure \ref{fig:notcheapinfo}, where we have substituted in \(y = 1 - x\) (thanks to Theorem \ref{thmcomparisonshop}).

With the equilibrium in hand, we can ask, how does consumer welfare change with the size of the information frictions? For intermediate costs, the consumer's learning is unaffected by the information cost; she learns the same for all \(\kappa\), and therefore the firms' pricing is the same. However, information is becoming more expensive, so the consumer's welfare is strictly decreasing in \(\kappa\) on this interval. 

On the flip side, when \(\kappa \geq \bar{\kappa}\) the opposite relationship exists. When \(\kappa\) increases, the consumer learns less. By learning less, the consumer induces more intense competition, which drives down the distribution of prices. The price effect dominates, so consumer welfare is increasing in the size of the information friction. 
\begin{proposition}\label{firstcomparativestat}
    For intermediate information costs, (\(\kappa 
 \in \left[\ubar{\kappa},\bar{\kappa}\right]\)), the consumer's welfare is strictly decreasing in the size of the friction. For large information costs, (\(\kappa 
  \geq \bar{\kappa}\)), the consumer's welfare is strictly increasing in the size of the friction.
\end{proposition}
\noindent The proof is in \(\S\)\ref{proof:firstcomparativestat}.

\subsection{Cheap Information}

\begin{figure}
    \centering
    \includegraphics[scale=.25]{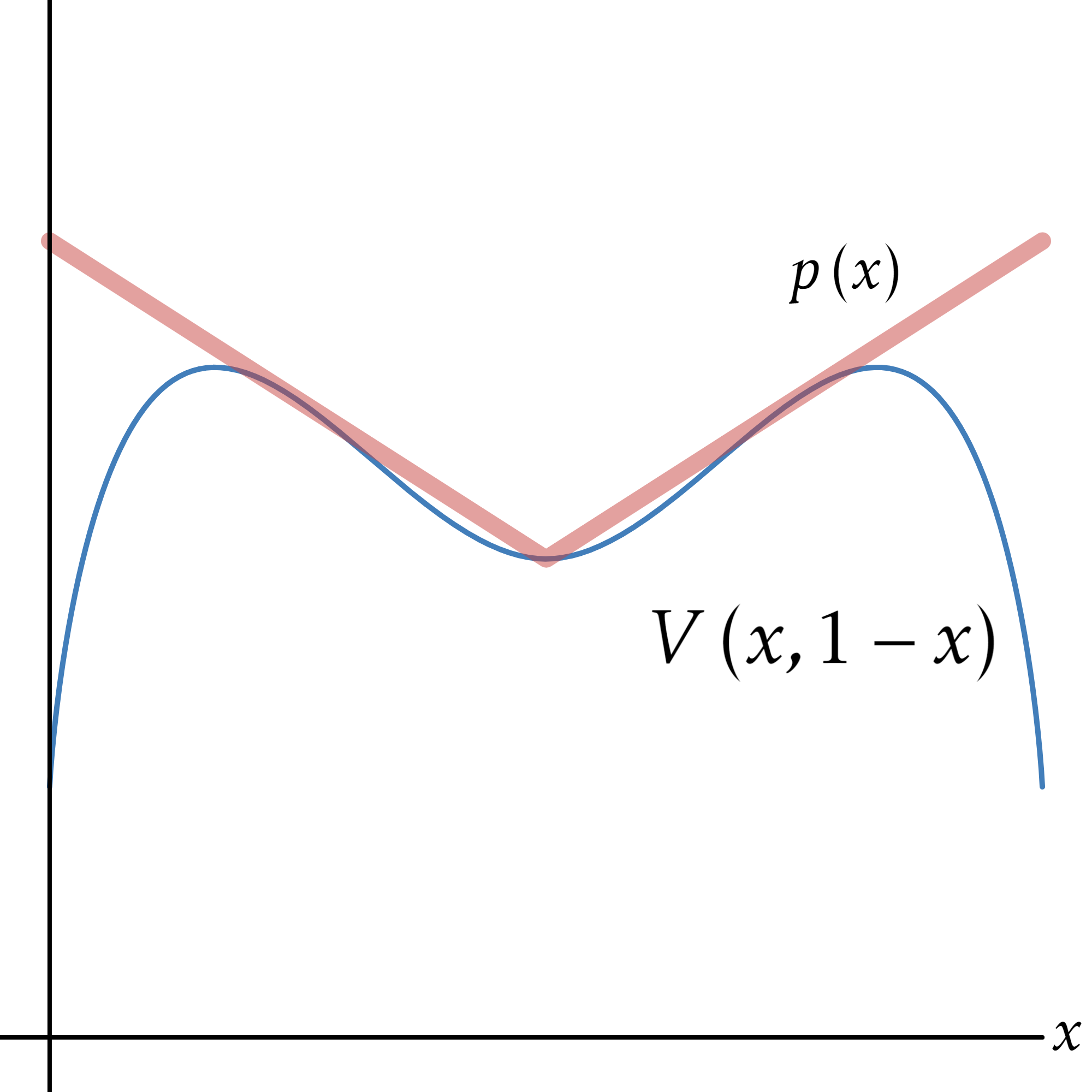}
    \caption{Cheap Information Equilibrium}
    \label{fig:cheapinfo}
\end{figure}

We say that a consumer \textbf{Comparison Shops With Occasional Indifference} if her acquired distribution over valuations has support on \(3\) points \(\left(\frac{1-\lambda}{2}, \frac{1+\lambda}{2}\right)\), \(\left(\frac{1}{2}, \frac{1}{2}\right)\), and \(\left(\frac{1+\lambda}{2}, \frac{1-\lambda}{2}\right)\).

\begin{theorem}\label{secondmaintheorem}
    If \(\kappa\) is sufficiently low there is an equilibrium in which the consumer comparison shops with occasional indifference and firms randomize over prices according to Expression \ref{priceset2}.
\end{theorem}

\begin{corollary}\label{efficientcorollary}
   As information costs vanish, \(\kappa \downarrow 0\), there is an efficient limiting equilibrium.
\end{corollary}
The proof for both are in \(\S\)\ref{proof:secondmaintheorem}. The limiting equilibrium has support on three points; \(\left(0,1\right)\), \(\left(1,0\right)\), and \(\left(\frac{1}{2},\frac{1}{2}\right)\).The firms price so that the consumer purchases from the advantaged firm--if there is an advantaged firm after her learning--with certainty. The consumer never purchases from a firm that is worse than the other \textit{ex post}. The consumer's learning in this equilibrium is depicted in Figure \ref{fig:cheapinfo}.

\section{Discussion and Extensions}\label{extensions}
In order to make traction, we needed to make several simplifying assumptions. Let us briefly explain a few of our assumptions. 

\smallskip

\noindent \textbf{Binary Values:} We assume that each firm's product takes just one of two values. When frictions are large, this is inconsequential: we show in \(\S\)\ref{sec:density} that unless \(\kappa\) is too small, there is an analogous equilibrium to the one we construct in our main specification when the consumer's value for the two firms' products is distributed according to some density on \(\left[0,1\right]^2\). It is when information is cheap that the problem changes: the learning with three-point support that we identify is no longer optimal for the consumer. Instead, we conjecture that the consumer now acquires a continuum of posteriors close to the prior plus possible point masses on more extreme posteriors. 

\smallskip

\noindent \textbf{Symmetric Firms:} Like the previous assumption, this is for tractability. The equilibrium of the pricing-only game becomes quite difficult to construct when firms are asymmetric.

\smallskip

\noindent \textbf{Parametric Assumption on the Prior:} We make a rather cryptic stipulation that the positive correlation between the consumer's values for the two firms' goods cannot be too high. This is again due to the challenges in constructing an equilibrium in the pricing game between the firms: it is much easier to construct an equilibrium in the pricing-only game when the probability of the ``tie'' belief (when the consumer's distribution has symmetric three-point support) is sufficiently high. Our parametric assumption, thus, guarantees this is true in the consumer's optimal learning when frictions are sufficiently low.

\smallskip

\noindent \textbf{Private Learning Before Trading:} We assume that the consumer learns before she observes the firms' prices and that; moreover, this learning is private.\footnote{\cite*{matvejka2012simple} study a related scenario in which firms set prices before the consumer learns.} The timing in our main environment (learning before trading) is realistic in many environments: in particular, learning about a service provider's reputation seems especially fitting. Our timing assumption is also that made by \hyperlink{RRS}{RRS}, which allows us to identify the effects of competition cleanly. In \(\S\)\ref{sec:observable}, we allow for public learning and show that a hold-up problem emerges, which leads to zero information acquisition by the consumer. 

\smallskip

\noindent \textbf{Cost Function:} The consumer's cost of acquiring information is a linear functional of the distribution over posteriors. We assume this posterior-separable form for tractability. Moreover, we specify that the convex function that is integrated has unbounded slope at the boundaries of the unit square. This is done to ensure an ``interior'' solution in the consumer's information acquisition problem. Importantly, this specification only makes our convergence result more difficult to attain: if the slope were bounded our results would go through with the modification that the results would no longer be limit results but would hold for sufficiently small positive \(\kappa\).

\smallskip

\noindent \textbf{Full Market Coverage:} Our main specification assumes that the consumer's outside option is negligible, so she always purchases from one of the firms. Clearly, when \(\kappa\) is large, the equilibrium we construct remains an equilibrium even with a (potentially) relevant outside option of, say, \(0\). Naturally, this is not the case as \(\kappa\) vanishes. 

\subsection{Extension to a Prior with a Density}\label{sec:density}

An exact analog of Theorem \ref{firstmaintheorem} holds when the consumer's valuations for the two products are symmetrically distributed with nonzero density \(h\) on the unit square and the consumer's utility is affine in her valuation for the purchased product and additively separable in her valuation, the price, and the cost of acquiring information (which is posterior-mean measurable).

That is, suppose each firm is selling a product whose value to the consumer is a random variable \(Z_i\) with full support on \(\left[0,1\right]\). Random vector \(\left(Z_1, Z_2\right)\) is distributed on \(\left[0,1\right]^2\) according to continuous density \(h\left(z_1, z_2\right)\), which is symmetric around the diagonal \(y=x\), i.e., \(h\left(z_1, z_2\right) = h\left(z_2, z_1\right)\) for all \(z_1, z_2 \in \left[0,1\right]\). The prior expected value is \(\frac{1}{2} = \int_{0}^{1}\int_{0}^{1}a f\left(a,b\right)dbda\). 

The consumer may acquire any fusion \(G \in \mathcal{F}_H\) of the prior at cost \(C\left(G\right) = \kappa \int c dG\) where we maintain the assumptions from the model section above. Then,
\begin{proposition}\label{firstmaintheoremanalog}
    If \(\kappa\) is sufficiently high, there is an equilibrium in which the consumer comparison shops with uniform two-point support and firms randomize over prices according to Expression \ref{2firmeqpricing}.
\end{proposition}

\noindent  The proof is in \(\S\)\ref{proof:firstmaintheoremanalog}. The comparative statics from Proposition \ref{firstcomparativestat} also carry over:
\begin{proposition}
    For intermediate information costs, (\(\kappa 
 \in \left[\ubar{\kappa},\bar{\kappa}\right]\)), the consumer's welfare is strictly decreasing in the size of the friction. For large information costs, (\(\kappa 
  \geq \bar{\kappa}\)), the consumer's welfare is strictly increasing in the size of the friction.
\end{proposition}
The intuition is also the same: in the intermediate-friction region, the consumer's optimal learning; and, therefore, the firms' behavior, stays the same as \(\kappa\) dwindles. The consumer accrues all of the benefits of cheaper information. When \(\kappa\) is large, the consumer's learning is affected and so firms raise their prices (on average) to take advantage of their greater market power. This (negative, for the consumer) force is dominant, and so cheaper information makes the consumer worse off.

\subsection{Observable Learning Benchmark}
\label{sec:observable}

A crucial assumption in our model is that learning is private. A natural comparison is the case in which the firms observe the consumer's acquired posterior \(\left(x, y\right) \in \left[0,1\right]^2\) before posting prices. The first step in characterizing the equilibrium is to characterize equilibria in the pricing-only game between the two firms for an arbitrary vector \(\left(x, y\right)\).
\begin{lemma}
\label{olbeq}
    For all \(\left(x, y\right)\), a pure-strategy equilibrium of the pricing-only game exists. If \(x < y\), there exists an equilibrium in which the vector of prices is \(\left(0,y-x\right)\) and the consumer always purchases from firm \(2\). For any \(\left(x,y\right)\) with \(x = y\), the unique equilibrium is the Bertrand outcome: both firms price at marginal cost, \(p_1 = p_2 = 0\).
\end{lemma}
It is straightforward to check that the equilibria constructed in Lemma \ref{olbeq} are particularly bad for the consumer unless \(x = y\): her expected payoff at any \(\left(x,y\right)\) with \(y \neq x\) is strictly negative. Moreover, it is clear that the consumer's net payoff at any \(\left(x,y\right)\), in any equilibrium, must be as follows:
\begin{lemma}\label{auxlemma}
    In any equilibrium of the pricing-only game with \(y > x\), the expected net payoff to the consumer is weakly less than \(x\).
\end{lemma}
The proof is in \(\S\)\ref{proof:auxlemma}. Working backward, we now conclude that the consumer will not learn.
\begin{proposition}\label{observbad}
    If \(\kappa > 0\), the unique equilibrium with observable learning is for the consumer to acquire no information: she chooses the degenerate distribution on the prior \(\left(\frac{1}{2}, \frac{1}{2}\right)\).
\end{proposition}
The proof is in \(\S\)\ref{proof:observbad}. Although this equilibrium is quite inefficient--no matter how cheap information is, the consumer does not learn--it is good for the consumer. By committing to not learn, she forces the firms to compete intensely and the consumer acquires all of the surplus in the market.

\section{Conclusion}
This paper develops a model of flexible information acquisition with imperfect competition between sellers. The buyer can purchase any signal about her valuation privately, but signals are costly. Without observing the buyer's learning strategy or outcome, sellers set prices and compete for the buyer. 

Technical difficulties arise because the equilibrium requires solving a multidimensional information design problem on top of an equilibrium pricing game that involves a distribution of prices. The value of posterior beliefs is endogenous and depends on the firms' pricing, which is random. We prove that the consumer only wants to learn the relative values, which we call comparison shopping.

Our main result proves that competition between sellers flips one inefficiency result of \hyperlink{RRS}{RRS}. With multiple sellers, as the cost of information vanishes, the equilibrium outcome is \textit{ex post} efficient: the consumer always purchases the higher-value product. We also do comparative statics: when the cost of information is high, consumer welfare decreases in the cost of information, but when information is cheap, this relationship flips. 


\bibliography{sample.bib}

\appendix

\section{Omitted Proofs and Derivations}

\subsection{Proposition \ref{monoppropproof} Proof}\label{proof:monopoly_eq}
\begin{proof}
    First, we argue that our construction constitutes an equilibrium. We need to show that a firm has no profitable deviation. The firm obtains a constant profit for any price \(p \in \left[\ubar{x},\bar{x}\right]\):
    \[\Pi\left(p\right) = p \left[1 - F\left(p\right)\right] = \ubar{x}\text{.}\]
 Its profit is strictly less than \(\ubar{x}\) for any price strictly below \(\ubar{x}\) and is \(0\) for any price strictly above \(\bar{x}\).
    
    It is also easy to check that the consumer has no profitable deviation: for any value \(x \in \left[\ubar{x},\bar{x}\right]\) the consumer's payoff is the affine function
    \[- \kappa c'\left(\ubar{x}\right) x + \kappa \left(\ubar{x}c'\left(\ubar{x}\right) - c\left(\ubar{x}\right)\right)\text{.}\]
    \textit{Viz.}, we have \[\int_{\ubar{x}}^{x}\left(x-p\right)dG\left(p\right) - \kappa c\left(x\right) = - \kappa c'\left(\ubar{x}\right) x + \kappa \left(\ubar{x}c'\left(\ubar{x}\right) - c\left(\ubar{x}\right)\right) \text{.}\]
    As the consumers payoff is continuous on the interior of \(\left(0,1\right)\), strictly concave for all \(x \in \left(0, \ubar{x}\right)\) and \(x \in \left(\bar{x}, 1\right)\), and linear on \(\left[\ubar{x},\bar{x}\right]\), we conclude that this distribution is also a best response for the consumer. Finally, Equation \ref{eq1} is pinned down by \(G\left(1\right) = 1\). 
    
    Uniqueness follows from the proof of Theorem 5.2 in \cite*{jainwhitflex2}.
\end{proof}
\subsection{Proposition \ref{monopcsproof} Proof}\label{proof:monopoly_cor}
\begin{proof}
    As \(\kappa \downarrow 0\), from Equation \ref{eq1}, we see that we must have \(\bar{x} \uparrow 1\) and/or \(\ubar{x} \downarrow 1\). We claim that we cannot have \(\ubar{x} \downarrow 0\). In fact, \(\ubar{x}\) must always be strictly above a strictly positive number as follows.
    \begin{claim}
    For all \(\kappa > 0\), \(\ubar{x} \geq a\), where \(a\) solves \[1 - \log{a} - \frac{\mu}{a} = 0\text{.}\] \end{claim}
\begin{proof}
Differentiating the LHS of Equation \ref{eq2} with respect to \(\bar{x}\), we get \(\frac{1}{\bar{x}}> 0 \). Differentiating the LHS of Equation \ref{eq2} with respect to \(\ubar{x}\), we get \(\frac{\mu}{\ubar{x}^2} - \frac{1}{\ubar{x}}> 0 \), as \(\ubar{x} < \mu\). Accordingly, by the IFT, \(\ubar{x}'\left(\bar{x}\right) < 0\). Thus, \(\ubar{x}\) is minimized when \(\bar{x} = 1\), which produces the specified equation. The unique solution to that equation \(a \in \left(0, \mu\right)\) is strictly increasing in \(\mu\) and is approximately \(.19\) when \(\mu = \frac{1}{2}\) and does not equal \(0\) for all \(\mu > 0\).
\end{proof}
Thus, we have shown that \(\ubar{x} \downarrow a\) and \(\bar{x} \uparrow 1\) as \(\kappa \downarrow 0\). Finally, pick an arbitrary \(p \in \left[\ubar{x}\left(0\right), \bar{x}\left(0\right)\right)\) and observe that 
\[\lim_{\kappa \downarrow 0} G\left(p\right) = \lim_{\kappa \downarrow 0}\kappa \left(c'\left(p\right) - c'\left(\ubar{x}\right)\right) = 0\text{,}\] and so we must have \(G \to \delta_{1}\) (in distribution) as \(\kappa \downarrow 0\). \end{proof}

\subsection{Lemma \ref{eq_rand_2types} Proof}\label{proof:eq_rand_2types}

\begin{proof} For a contradiction, assume a symmetric pure-strategy equilibrium exists. Then a firm's demand is (locally) perfectly inelastic--if it raises its price slightly, the consumer will purchase from it with the same probability. This deviation yields strictly higher profits, which contradicts our original pure-strategy being an equilibrium.\end{proof} 

\subsection{Proposition \ref{eqdist} Proof}\label{proof:eqdist}

\begin{proof}
We are looking for a symmetric equilibrium in which each firm chooses an atomless distribution over prices \(\Gamma\left(p\right)\) with support on \(\left[\ubar{p}, \ubar{p}+2\lambda\right]\). We guess further that \(\Gamma\) can be written as \(\Gamma\left(p\right) = \Gamma_L\left(p\right)\) for \(p \in \left[\ubar{p}, \ubar{p}+\lambda\right]\) and \(\Gamma\left(p\right) = \Gamma_H\left(p\right)\) for \(p \in \left[\ubar{p} + \lambda, \ubar{p}+2\lambda\right]\). The profit for a firm is
\[\Pi\left(p\right) = \begin{cases}
\frac{p}{2}\left[2-\Gamma_{H}\left(p+\lambda\right) \right], \quad &\ubar{p} \leq p \leq \ubar{p} + \lambda\\
\frac{p}{2}\left(1-\Gamma_{L}\left(p-\lambda\right)\right), \quad &\ubar{p} + \lambda \leq p \leq \ubar{p} + 2\lambda
\end{cases} \text{ .}\]
For any on-path \(p\) a firm's payoff must equal some constant \(k\). Imposing this and the conditions for \(\Gamma\) to be a cdf, we get the functional form specified in Expression \ref{2firmeqpricing}.

Finally, we need to verify that firms do not want to
choose a price outside of the conjectured region. If a firm chooses a price \(p \in \left[\ubar{p}-\lambda,\ubar{p}\right]\), its payoff is
\[\frac{p}{2}\left[2-\Gamma_{L}\left(p+\lambda\right)\right] = \frac{p}{2}\left[\left(\frac{\sqrt{2}\lambda + \lambda}{p+2\lambda}\right) + 1\right]\text{ .}\]
The derivative of this with respect to \(p\) is
\[\frac{2\lambda\left(\frac{\sqrt{2}\lambda+\lambda}{p+2\lambda}\right)}{p+2\lambda} + 1 > 0 \text{ ,}\]
whence we conclude a firm does not want to deviate to a price in this region (we have implicitly assumed that \(\ubar{p} \geq \lambda\), but this is fine since a firm obviously does not want to deviate to a negative price). Evidently, if a firm chooses any price \(p \leq \ubar{p} - \lambda\) its payoff is just \(p\), which is obviously strictly increasing in \(p\) and hence equals \(\ubar{p} + \lambda\), which we just established is not an improvement for the firm. The last case is that in which a firm chooses a price \(p \in \left[\ubar{p} + 2\lambda, \ubar{p} + 3\lambda\right]\). In that case, a firm's profit is
\[\frac{p}{2}\left(1-\Gamma_{H}\left(p-\lambda\right)\right) = \frac{p}{2}\left(\frac{\sqrt{2}\lambda+3\lambda-p}{p-2\lambda}\right) \text{ ,}\]
which is strictly decreasing in \(p\).

The uniqueness argument is analogous to that argued for the atomless equilibrium of Proposition 3 in \cite*{moscarini2001}.\end{proof}

\subsection{Lemma \ref{eq_rand_3types} Proof}\label{proof:eq_rand_3types}
\begin{proof}
Because the consumer is indifferent between the two firms with strictly positive probability, a standard under-cutting argument eliminates any symmetric equilibria in which a firm sets some price with strictly positive probability.
\end{proof}

\subsection{Proposition \ref{eq_dist_3types} Proof}\label{proof:eq_dist_3types}
\begin{proof}
If one firm, firm \(2\), say, chooses \(\Phi\), firm \(1\)'s profit as a function of \(p\) is
\(\left(1-q\right)\frac{q}{1-2q}\lambda\), a constant, for all \(p \in \left[\frac{q}{1-2q}\lambda, \frac{q}{1-2q}\lambda + \lambda\right]\). For all \(p \in \left[\frac{q}{1-2q}\lambda + \lambda, \frac{q}{1-2q}\lambda + 2 \lambda\right]\), firm \(1\)'s profit is \[p q \left(1-\Phi\left(p-\lambda\right)\right) \text{,}\] which is strictly decreasing in \(p\). For all \(p \geq \frac{q}{1-2q}\lambda + 2 \lambda\) firm \(1\)'s profit is \(0\). Finally, for all \(p \in \left[0, \frac{q}{1-2q}\lambda + \lambda\right]\), firm \(1\)'s profit is \(p \left(1-q \Phi\left(p+\lambda\right)\right)\). For all \[q \leq \frac{\frac{\sqrt[3]{9\sqrt{93}-47}}{\sqrt[3]{2}}-\frac{11\sqrt[3]{2}}{\sqrt[3]{9\sqrt{93}-47}}+5}{9} \approx .406 \text{,}\] this function is strictly increasing on this interval. \end{proof}

\subsection{Theorem \ref{thmcomparisonshop} Proof}\label{proof:thmcomparisonshop}
\begin{proof}
    By symmetry we restrict attention WLOG to the case \(y \geq x\). We assume that the firms each choose the distributions over prices \(F\) with support on \(\left[\ubar{p}, \bar{p}\right]\). Defining \(\lambda \coloneqq \bar{p} - \ubar{p}\), the consumer's payoff from posterior \(\left(x,y\right)\), is 
    \[\label{imp2}\tag{\Uranus}V\left(x,y\right) = \begin{cases}
y - \mathbb{E}\left[p\right]-\kappa c\left(x,y\right), \quad &y \geq x + \lambda\\
y - \mathbb{E}\left[p\right] + U\left(z\right) - \kappa c\left(x,y\right), \quad &x+\lambda \geq y \geq x
\end{cases} \text{ ,}\]
where \(z \coloneqq y-x\) and
\[\mathbb{E}\left[p\right] = \int_{\ubar{p}}^{\ubar{p}+\lambda}pdF\left(p\right) \text{ ,}\] and
\[U\left(z\right) \coloneqq  \int_{\ubar{p}+z}^{\ubar{p}+\lambda}\left(p-z\right)F\left(p-z\right)dF\left(p\right) - \int_{\ubar{p}}^{\ubar{p}+\lambda-z}p\left(1-F\left(p+z\right)\right)dF\left(p\right) \text{ .}\]
Directly,
\[V_{xx}\left(x,y\right) = \begin{cases}
-\kappa c_{xx}\left(x,y\right),\\
\int_{\ubar{p}+z}^{\ubar{p}+\lambda}f\left(p-z\right)dF\left(p\right) - \kappa c_{xx}\left(x,y\right)
\end{cases} \quad 
V_{yy}\left(x,y\right) = \begin{cases}
-\kappa c_{yy}\left(x,y\right),\\
\int_{\ubar{p}+z}^{\ubar{p}+\lambda}f\left(p-z\right)dF\left(p\right) - \kappa c_{yy}\left(x,y\right)
\end{cases} \text{ ,}\]
and
\[V_{xy} = - \int_{\ubar{p}+z}^{\ubar{p}+\lambda}f\left(p-z\right)dF\left(p\right) - \kappa c_{xy}\left(x,y\right)\text{.}\]
The directional second derivative in the direction of \(\left(1,1\right)\) is \(-\kappa c_{xx}\left(x,y\right) - \kappa c_{yy}\left(x,y\right) - 2 \kappa c_{xy}\left(x,y\right) < 0\), by the strict convexity of \(c\).

Let us now evaluate the function \(c\left(x, a+x\right)\), where \(a\) is a parameter taking values in \(\left[-1,1\right]\). We have already (just) shown that it is strictly concave. Directly,
\[\frac{\partial}{dx} c\left(x,a+x\right) = \varphi'\left(x\right) - \varphi'\left(1-x\right) + \varphi'\left(a+x\right) - \varphi'\left(1-a-x\right)\text{.}\] By direct substitution, we see that this equals \(0\) when \(y = 1-x\), i.e., \(x = \frac{1-a}{2}\). Accordingly, for all \(a\) in the specified interval, \(c\left(x,a+x\right)\) is maximized at \(y = 1-x\); \textit{viz.}, on the the comparison shopping line. Thus, any price function of the value function restricted to the comparison shopping line, \(V\left(x, 1-x\right)\), must correspond to a price function that lies everywhere weakly above \(V\left(x,y\right)\), and so learning along the comparison-shopping line is optimal. \end{proof}

\subsection{Theorem \ref{firstmaintheorem} Proof}\label{proof:firstmaintheorem}
\begin{proof}
For convenience, define \(\ubar{p} \coloneqq \sqrt{2} \lambda\), \(\tilde{p} \coloneqq \ubar{p} + \lambda\), and \(\bar{p} \coloneqq \ubar{p} + 2\lambda\). The consumer's payoff as a function of the realized posterior belief \(\left(x,y\right)\) is (restricting attention to \(y \geq x\) by symmetry)
\[\label{imp}\tag{\Neptune}V\left(x,y\right) = \begin{cases}
y - \mathbb{E}\left[p\right]-\kappa c\left(x,y\right), \quad &y \geq x + 2 \lambda\\
y - \mathbb{E}\left[p\right] + T_1\left(z\right) - \kappa c\left(x,y\right), \quad &x+2 \lambda \geq y \geq x + \lambda\\
y - \mathbb{E}\left[p\right] + T_2\left(z\right) - \kappa c\left(x,y\right), \quad &x+\lambda \geq y \geq x
\end{cases} \text{ ,}\]
where \(z \coloneqq y-x\) and
\[\mathbb{E}\left[p\right] = \int_{\ubar{p}}^{\Tilde{p}}pd\Gamma_{L}\left(p\right) +  \int_{\Tilde{p}}^{\bar{p}}pd\Gamma_{L}\left(p\right) = \left(\left(\sqrt{2}+1\right)\log\left(\sqrt{2}+1\right)+\sqrt{2}-1\right)\lambda \text{ ,}\]
\[T_1\left(z\right) \coloneqq  \int_{\ubar{p}}^{\bar{p}-z}\left(1-\Gamma_{H}\left(p+z\right)\right)\Gamma_{L}\left(p\right)dp \text{ ,}\]
and
\[T_2\left(z\right) \coloneqq  \int_{\tilde{p}}^{\bar{p}-z}\left(1-\Gamma_{H}\left(p+z\right)\right)\Gamma_{H}\left(p\right)dp + \int_{\tilde{p}-z}^{\tilde{p}}\left(1-\Gamma_{H}\left(p+z\right)\right)\Gamma_{L}\left(p\right)dp+\int_{\ubar{p}}^{\tilde{p}-z}\left(1-\Gamma_{L}\left(p+z\right)\right)\Gamma_{L}\left(p\right)dp \text{ .}\]

From Theorem \ref{thmcomparisonshop}, we may restrict attention to learning along the line \(y = 1-x\). The directional derivative along vector \(\left(1,-1\right)\), evaluated at all points of the form \(\left(x,1-x\right)\), is
\[\tag{\Pluto}\label{directional}D\left(x\right) = \begin{cases}
-\kappa c_x\left(x, 1-x\right) + \kappa c_y\left(x, 1-x\right) - 1, \quad &x \leq \frac{1}{2} - \lambda\\
2P_1\left(1-2x\right)-\kappa c_x\left(x, 1-x\right) + \kappa c_y\left(x, 1-x\right)  - 1, \quad &\frac{1}{2} - \lambda \leq x \leq \frac{1 - \lambda}{2}\\
2P_2\left(1-2x\right)-\kappa c_x\left(x, 1-x\right) + \kappa c_y\left(x, 1-x\right) - 1, \quad &\frac{1 - \lambda}{2} \leq x \leq \frac{1}{2}
\end{cases} \text{ ,}\]
where
\[P_1\left(z\right) \coloneqq \int_{\ubar{p}}^{\bar{p}-z}\gamma_{H}\left(p+z\right)\Gamma_{L}\left(p\right)dp \text{ ,}\]
and
\[P_2\left(z\right) \coloneqq \int_{\tilde{p}}^{\bar{p}-z}\gamma_{H}\left(p+z\right)\Gamma_{H}\left(p\right)dp + \int_{\tilde{p}-z}^{\tilde{p}}\gamma_{H}\left(p+z\right)\Gamma_{L}\left(p\right)dp+\int_{\ubar{p}}^{\tilde{p}-z}\gamma_{L}\left(p+z\right)\Gamma_{L}\left(p\right)dp \text{ .}\]

Direct substitution yields \(D\left(\frac{1}{2} \right) = 0\). Moreover, by the symmetry and convexity of \(c\), and since \(c\left(\frac{1}{2}, \frac{1}{2}\right) = 0\), for \(x \leq \frac{1}{2}\), \(c_{y}\left(x, 1-x\right) - c_{x}\left(x, 1-x\right) \geq 0\) with equality at \(x = \frac{1}{2}\).

When \(\kappa\) is large or when \(\omega\) is sufficiently large, we need \(D\left(\frac{1 - \lambda}{2}\right) = D\left(\frac{1+\lambda}{2}\right) = 0\), i.e., 
\[\label{sun}\tag{\Sun}\tau\left(\kappa, \lambda\right) \coloneqq 2P_1\left(\lambda\right)-\kappa c_x\left(\frac{1 - \lambda}{2}, \frac{1+\lambda}{2}\right) + \kappa c_y\left(\frac{1 - \lambda}{2}, \frac{1+\lambda}{2}\right)  - 1 = 0\text{.}\]
Note that
\[P_1\left(\lambda\right) = \int_{\ubar{p}}^{\tilde{p}}\gamma_{H}\left(p+\lambda\right)\Gamma_{L}\left(p\right)dp = -\frac{\left(2^\frac{5}{2}+6\right)\log\left(\sqrt{2}+2\right)-\left(2^\frac{7}{2}+12\right)\log\left(\sqrt{2}+1\right)+\left(2^\frac{3}{2}+3\right)\log\left(2\right)+2}{2} \text{,}\]
which is evidently independent of the parameters and is approximately \(\frac{1}{10}\). Directly, \(\tau'\left(\kappa\right) > 0\) and 
\[\tau'\left(\lambda\right) = \frac{\kappa}{2}\left[c_{xx}\left(\frac{1 - \lambda}{2}, \frac{1+\lambda}{2}\right) - 2 c_{xy}\left(\frac{1 - \lambda}{2}, \frac{1+\lambda}{2}\right) + c_{yy}\left(\frac{1 - \lambda}{2}, \frac{1+\lambda}{2}\right) \right] > 0 \text{.}\]
By the implicit function theorem \(\lambda'\left(\kappa\right) < 0\). Moreover, \(\lim_{\lambda \to 1}\tau = \infty\) and \(\tau\left(\kappa, 0\right) < 0\). Accordingly, there is a unique solution \(\lambda^* = \lambda\left(\kappa\right)\) to this equation, which is strictly decreasing in \(\kappa\). Moreover, \(\lim_{\kappa \uparrow \infty} \lambda^*\left(\kappa\right) = 0\) and \(\lim_{\kappa \downarrow 0} \lambda^*\left(\kappa\right) = 1\). We also need to check the following:
\begin{claim}\label{claim51}
\(D\left(x\right) \leq 0\) for all \(x \in \left[\frac{1}{2} - \lambda^*, \frac{1}{2}\right]\); and \(V\left(\frac{1 - \lambda}{2}, \frac{1+\lambda}{2}\right) \geq V\left(\frac{1}{2}, \frac{1}{2}\right)\).
\end{claim}
\begin{proof}
Directly, we differentiate the function \(D\) (from Expression \ref{directional})  with respect to \(x\). This yields \[D'\left(x\right) = \underbrace{-\kappa c_{xx}\left(x, 1-x\right) - \kappa c_{yy}\left(x, 1-x\right)}_{\eqqcolon \kappa \rho(x) < 0} + \tau\left(x\right)\text{,}\]
where
\[\tau\left(x\right) = \begin{cases}
0, \quad &x \leq \frac{1}{2} - \lambda\\
4R\left(1-x\right), \quad &\frac{1}{2} - \lambda \leq x \leq \frac{1 - \lambda}{2}\\
4M\left(1-x\right), \quad &\frac{1 - \lambda}{2} \leq x \leq \frac{1}{2}
\end{cases} \text{ ,}\]
where \[R\left(z\right) \coloneqq \int_{\ubar{p}}^{\bar{p}-z}\gamma_H\left(p+z\right)\gamma_{L}\left(p\right)dp \text{ ,}\]
and
\[M\left(z\right) \coloneqq \int_{\tilde{p}}^{\bar{p}-z}\gamma_H\left(p+z\right)\gamma_{H}\left(p\right)dp + \int_{\tilde{p}-z}^{\tilde{p}}\gamma_H\left(p+z\right)\gamma_{L}\left(p\right)dp + \int_{\ubar{p}}^{\tilde{p}-z}\gamma_L\left(p+z\right)\gamma_{L}\left(p\right)dp \text{ .}\]
It is straightforward to check that \(\tau\left(x\right) \geq 0\) (strictly if \(x > \frac{1}{2} - \lambda\)) and that it is strictly increasing in \(x\) (for all \(x > \frac{1}{2} - \lambda\)). Moreover, \(\rho'\left(x\right) = c_{yyy} - c_{xxx} \geq 0\), by assumption. Accordingly, \(D'\) has at most one sign change, from negative to positive. Given the zero slope condition at \(x = \frac{1-\lambda^*}{2}\), this establishes the claim.\end{proof}

For this distribution to be feasible (a fusion of the prior) we need \(\frac{\lambda^*}{2} \leq \omega\). Define \(\bar{\kappa} \geq 0\) to be the value of \(\kappa\) such that \(\lambda^* = 2 \omega\). Observe that \(\bar{\kappa} = 0\) if and only if \(\omega = \frac{1}{2}\), which is the perfect negative correlation case. Directly, \(D\left(\frac{1}{2} - \omega\right)\) is strictly increasing in \(\kappa\).

\begin{claim}
    There exists an interval of \(\kappa\)s, \(\left[\ubar{\kappa},\bar{\kappa}\right]\), where \(\ubar{\kappa} \in \left[0, \bar{\kappa}\right]\) for which the equilibrium \(\lambda = 2 \omega\).
\end{claim}
\begin{proof}
    Directly, \(\frac{\partial}{\partial \kappa} D\left(x\right) > 0\). By construction, when \(\kappa = \bar{\kappa}\), the equilibrium \(\lambda^* = 2 \omega\). Moreover, as we noted in Claim \ref{claim51}, \(V\left(\frac{1 - \lambda}{2}, \frac{1+\lambda}{2}\right) \geq V\left(\frac{1}{2}, \frac{1}{2}\right)\). If this is an equality then \(\ubar{\kappa} = \bar{\kappa}\). If this inequality is strict then, by the intermediate value theorem, there is an interval of \(\kappa\)s (\(\left[\ubar{\kappa}, \bar{\kappa}\right]\)) for which line tangent to \(V\left(x-\omega, x+\omega\right)\) lies above \(V\left(x, 1-x\right)\) at \(\frac{1}{2}\).
\end{proof}
\end{proof}

\subsection{Proposition \ref{firstcomparativestat}}\label{proof:firstcomparativestat}
\begin{proof}
    Thanks to the discussion in the text, we need only consider the case \(\kappa \geq \bar{\kappa}\). In this case, the consumer's payoff at equilibrium is 
    \[\frac{1}{2} + \underbrace{\frac{\lambda}{2} - \mathbb{E}\left[p\right] + \int_{\ubar{p}}^{\tilde{p}}\left(1-\Gamma_{H}\left(p+\lambda\right)\right)\Gamma_{L}\left(p\right)dp}_{S\left(\lambda\right)} - \kappa c\left(\frac{1 - \lambda}{2},\frac{1+\lambda}{2}\right)\text{.}\]
    The derivative of this with respect to \(\lambda\) is \(S'\left(\lambda\right) - \frac{1}{2}\kappa\left(c_y\left(\frac{1 - \lambda}{2},\frac{1+\lambda}{2}\right) - c_x\left(\frac{1 - \lambda}{2},\frac{1+\lambda}{2}\right)\right)\). However, by the Optimality Equation \ref{sun}, \(c_y\left(\frac{1 - \lambda}{2},\frac{1+\lambda}{2}\right) - c_x\left(\frac{1 - \lambda}{2},\frac{1+\lambda}{2}\right) = 2 P_1\left(\lambda\right) - 1\). Summing everything up and simplifying, we obtain \(-1 - \sqrt{2} < 0\).
\end{proof}

\subsection{Theorem \ref{secondmaintheorem} and Corollary \ref{efficientcorollary} Proofs}\label{proof:secondmaintheorem}
\begin{proof}
    For convenience, define \(\ubar{p} \coloneqq \frac{q}{1-2q}\lambda\). The consumer's payoff as a function of the realized posterior belief \(x\) is (restricting attention to \(y \geq x\) by symmetry)\[V\left(x,y\right) = \begin{cases}
y - \mathbb{E}\left[p\right]-\kappa c\left(x,y\right), \quad &y \geq x + \lambda\\
y - \mathbb{E}\left[p\right] + U\left(z\right) - \kappa c\left(x,y\right), \quad &x+\lambda \geq y \geq x
\end{cases} \text{ ,}\]
where \(z \coloneqq y-x\) and
\[\mathbb{E}\left[p\right] = \int_{\ubar{p}}^{\ubar{p}+\lambda}pd\Phi\left(p\right) \text{ ,}\] and
\[U\left(z\right) \coloneqq  \int_{\ubar{p}+z}^{\ubar{p}+\lambda}\left(p-z\right)\Phi\left(p-z\right)d\Phi\left(p\right) - \int_{\ubar{p}}^{\ubar{p}+\lambda-z}p\left(1-\Phi\left(p+z\right)\right)d\Phi\left(p\right) \text{ .}\]

For this to be an equilibrium, we need for there to be a line \(\alpha x + \beta\) lying everywhere above \(V\left(x, 1-x\right)\) on \(0 \leq x \leq \frac{1}{2}\), and intersecting \(V\left(x, 1-x\right)\) at \(\frac{1 - \lambda}{2}\) and \(\frac{1}{2}\). Removing \(\mathbb{E}\left[p\right]\) since it is a constant, we compute 
\[V\left(\frac{1}{2}, \frac{1}{2}\right) = \frac{1}{2} -\frac{\lambda\left(1-q\right)q \left(\log\left(\frac{q}{1-q}\right)-4q+2\right)}{\left(1-2q\right)^{3}}\text{.}\] Moreover, \(\alpha = \kappa c_y\left(\frac{1 - \lambda}{2}, \frac{1+\lambda}{2}\right) - \kappa c_x\left(\frac{1 - \lambda}{2}, \frac{1+\lambda}{2}\right)-1\).

We need \[\left(\kappa c_y\left(\frac{1 - \lambda}{2}, \frac{1+\lambda}{2}\right) - \kappa c_x\left(\frac{1 - \lambda}{2}, \frac{1+\lambda}{2}\right)-1\right) \left(\frac{1 - \lambda}{2}\right) + \beta = \frac{1+\lambda}{2} - \kappa c\left(\frac{1 - \lambda}{2}, \frac{1+\lambda}{2}\right)\text{,}\]
or
\[\beta = 1 - \kappa c\left(\frac{1 - \lambda}{2}, \frac{1+\lambda}{2}\right) - \left(\kappa c_y\left(\frac{1 - \lambda}{2}, \frac{1+\lambda}{2}\right) - \kappa c_x\left(\frac{1 - \lambda}{2}, \frac{1+\lambda}{2}\right)\right) \left(\frac{1 - \lambda}{2}\right)\text{.}\]
We also need 
\[\left(\kappa c_y\left(\frac{1 - \lambda}{2}, \frac{1+\lambda}{2}\right) - \kappa c_x\left(\frac{1 - \lambda}{2}, \frac{1+\lambda}{2}\right)-1\right) \frac{1}{2} + \beta = \frac{1}{2} -\frac{\lambda\left(1-q\right)q \left(\log\left(\frac{q}{1-q}\right)-4q+2\right)}{\left(1-2q\right)^{3}}\text{,}\]
or
\[\beta = 1 -\frac{\lambda\left(1-q\right)q \left(\log\left(\frac{q}{1-q}\right)-4q+2\right)}{\left(1-2q\right)^{3}} - \frac{1}{2}\left(\kappa c_y\left(\frac{1 - \lambda}{2}, \frac{1+\lambda}{2}\right) -  \kappa c_x\left(\frac{1 - \lambda}{2}, \frac{1+\lambda}{2}\right)\right) \text{.}\]
Equating the \(\beta\)s, we get
\[\tag{\Moon}\label{pinninglambda}\omega \underbrace{\frac{\left(1-q\right) \left(\log\left(\frac{q}{1-q}\right)-4q+2\right)}{\left(1-2q\right)^{3}}}_{\eqqcolon t\left(\lambda\right)} + \kappa\underbrace{\left[\lambda d'\left(\lambda\right) - d\left(\lambda\right)\right]}_{v\left(\lambda\right)} = 0\text{,}\]
where
\[d\left(\lambda\right) \coloneqq c\left(\frac{1 - \lambda}{2}, \frac{1+\lambda}{2}\right)\text{,}\]
and where we used the fact that \(t \equiv t\left(\lambda\right)\) (as \(q = \frac{\omega}{\lambda}\)).

Furthermore, note that we must have \(2 \omega \leq \lambda \leq 1\). Directly, \(t'\left(\lambda\right) < 0\) and \(t\left(\lambda\right) < 0\) for all \(\lambda \in \left[2\omega,1\right]\), and \(\lim_{\lambda \downarrow 2 \omega} t\left(\lambda\right) = 0\). Moreover, by the strict convexity of \(c\),  \(v\left(\lambda\right) > 0\). Likewise, \(v'\left(\lambda\right) = \lambda d''\left(\lambda\right) > 0\), and \(\lim_{\lambda \uparrow 1}v'\left(\lambda\right) = \infty\). Continuing along these lines, it is easy to compute that \(t'\left(\lambda\right)\) is bounded for all \(\lambda \in \left(2\omega, 1\right]\). Finally, it is straightforward to check that \(\lim_{\lambda \uparrow 1}v\left(\lambda\right) = \infty\).

From the observations in the previous paragraph, we conclude the following:
\begin{enumerate}[label={(\roman*)},noitemsep,topsep=0pt]
    \item If \(\kappa\) is sufficiently small, then a unique solution \(\lambda^*\left(\kappa\right)\) to Equation \ref{pinninglambda} exists.
    \item In this unique solution, \(\lambda^*\) is strictly decreasing in \(\kappa\).
    \item As \(\kappa \downarrow 0\), \(\lambda^* \uparrow 1\).
\end{enumerate}
This last item is the stated corollary (\ref{efficientcorollary}).\end{proof}
\subsection{Proposition \ref{firstmaintheoremanalog} Proof}\label{proof:firstmaintheoremanalog}
\begin{proof}
    This result is an immediate implication of the fact that in the proof of Theorem \ref{firstmaintheorem}, we show that as \(\kappa\) increases, \(\lambda\) decreases and in the limit goes to \(0\). For any prior as specified in this section, there is a threshold \(\lambda > 0\) such that the comparison shopping with uniform two-point support distribution is a fusion of the prior. That we can use the same price-function approach follows from \cite*{persuasionduality}. Alternatively, as \cite*{kleiner2023extreme} establish, these distributions are exposed (in their parlance, ``strongly exposed'') points in the set of finitely-supported fusions of the prior. When frictions are high (\(\kappa \geq \bar{\kappa}\)), the associated power diagram is the trivial power diagram consisting of a single element (\(\left[0,1\right]^2\)). When frictions are moderate (\(\kappa \in \left[\ubar{\kappa},\bar{\kappa}\right]\)), the associated power diagram is convex partitional, with two elements, the triangles \(\Delta^1\) and \(\Delta^2\) (Expressions \ref{delta1} and \ref{delta2}).
\end{proof}
\subsection{Lemma \ref{auxlemma} Proof}\label{proof:auxlemma}
\begin{proof}
    Consider an arbitrary equilibrium and let \(v \geq 0\) be the infimum of the support of firm \(1\)'s distribution over prices. Naturally, then, \(v + \lambda\) must be the infimum of the support of firm \(2\)'s distribution over prices. Thus, the consumer's net payoff is weakly less than \(\max\left\{y-v-\lambda, x-v\right\} = x- v \leq x\).
\end{proof}
\subsection{Proposition \ref{observbad} Proof}\label{proof:observbad}
\begin{proof}
    By Lemma \ref{auxlemma}, the consumer's payoff at any \(\left(x,y\right)\) is bounded above by \(\min\left\{x,y\right\}\), which is weakly concave. For \(\kappa > 0\) this function is strictly concave.
\end{proof}

\end{document}